\documentclass[letterpaper, 10 pt, conference]{ieeeconf}  

\IEEEoverridecommandlockouts                              

\usepackage{amsfonts,latexsym,amsmath,amssymb}
\usepackage[mathscr]{eucal}
\usepackage{graphicx,color}
\usepackage{mathtools}
\usepackage{stmaryrd}
\usepackage{url}
\usepackage{dsfont}
 
\usepackage{amsthm}
\usepackage{bbm}
\usepackage{psfrag}

\usepackage{amsgen}
\usepackage{amstext}
\usepackage{amsbsy}
\usepackage{amsopn}

\usepackage{bm}

\usepackage{subcaption} 

\providecommand{\N}{\mathbb{N}}
\providecommand{\R}{\mathbb{R}}

\providecommand{\S}{\mathrm{S}}

\providecommand{\SO}{\mathbf{SO}}

\usepackage{accents}
\makeatletter
\providecommand{\scirc}{%
    \hbox{\fontfamily{\rmdefault}\fontsize{0.4\dimexpr(\f@size pt)}{0}\selectfont{\raisebox{-0.52ex}[0ex][-0.52ex]{$\circ$}}}}

\makeatother

\mathchardef\mhyphen="2D

\def \S {{\mathrm S}}

\usepackage[colorlinks,bookmarksopen,bookmarksnumbered,citecolor=red,urlcolor=red]{hyperref}

\makeatletter
\newcommand{\pushright}[1]{\ifmeasuring@#1\else\omit\hfill$\displaystyle#1$\fi\ignorespaces}
\newcommand{\pushleft}[1]{\ifmeasuring@#1\else\omit$\displaystyle#1$\hfill\fi\ignorespaces}
\makeatother

\newtheorem{theorem}{Theorem}
\newtheorem{proposition}{Proposition}
\newtheorem{definition}{Definition}
\newtheorem{assumption}{Assumption}
\newtheorem{lemma}{Lemma}
\newtheorem{remark}{Remark}

\usepackage{siunitx}

\title{\LARGE \bf
Leader-Follower Formation Control with Prescribed Convergence Rates under Bearing Persistence of Excitation (Extended Version)
}

\author{
Tarek Bouazza, Zhiqi Tang, \IEEEmembership{Member, IEEE}, Soulaimane Berkane, \IEEEmembership{Senior Member, IEEE}, \\ and Tarek Hamel, \IEEEmembership{Fellow, IEEE}
\thanks{This research is supported by the ``Grands Fonds Marins'' Project Deep-C, the Natural Sciences and Engineering Research Council of Canada (NSERC) under the Discovery Grant RGPIN-2020-04759, and the Fonds de recherche du Québec (FRQ).}
\thanks{
T. Bouazza is with I3S, CNRS, Université Côte d'Azur, 06903 Sophia Antipolis, France, {\tt\small bouazza@i3s.unice.fr}.}
\thanks{
Z. Tang is with the Department of Electrical and Electronic Engineering, the University of Manchester, Manchester M13 9PL, United Kingdom, {\tt\small zhiqi.tang@manchester.ac.uk}.}
\thanks{S. Berkane is with the Département d’informatique et d’ingénierie, Université du Québec en Outaouis, Gatineau, QC J8X 3X7, Canada, and also with the Department of Electrical
Engineering, Lakehead University, Thunder Bay, ON P7B 5E1, Canada,
{\tt\small soulaimane.berkane@uqo.ca}.}
\thanks{
T. Hamel is with I3S, CNRS, Université Côte d'Azur, 06903 Sophia Antipolis, France, and also with Institut Universitaire de France (IUF), 75005, Paris, France,
         {\tt\small thamel@i3s.unice.fr}.}
}

\begin{document}

\maketitle
\thispagestyle{empty}
\pagestyle{empty}


\begin{abstract}
This paper addresses leader-follower formation control using only relative bearing and velocity measurements. Bearing-based leader-follower control strategies commonly use fixed control gains, for which the guaranteed convergence rates explicitly depend on the persistence of excitation (PE) properties of the desired formations. We propose a time-varying matrix gain that evolves according to the bearing information available to each follower and decouples the convergence rate from the PE bound. We establish the well-posedness and uniform boundedness of the proposed gain for bearing-persistently exciting formations, and characterize the exponential convergence rate through a tunable design parameter. The proposed control design is further extended to leader-follower formations with double-integrator dynamics. Simulations illustrate the resulting convergence properties compared with fixed gain designs.
\end{abstract}

    \section{Introduction}

The problem of formation control has been extensively studied over the past two decades by both the robotics and control communities. Existing approaches are commonly classified into four main categories \cite{oh2015survey}: position-based formation control \cite{ren2007distributed,brinon2014cooperative}, displacement-based formation control \cite{ren2004coordination}, distance-based formation control \cite{anderson2007control}, and, more recently, bearing-based formation control \cite{basiri2010distributed}. Bearing-based methods have attracted growing interest because they require only relative bearing measurements, which can be obtained with lightweight, low-cost sensors such as monocular cameras, rather than inter-agent distances or global positions. 

Early work on bearing-based formation control focused mainly on planar formations and on controlling subtended bearing angles in each agent’s local frame (\textit{e.g.,} \cite{basiri2010distributed, bishop2011very}). The subsequent development of bearing rigidity theory provided conditions under which a formation is uniquely determined, up to translation and scale, by inter-agent bearings. Bearing rigidity in the plane (also termed parallel rigidity) was studied in \cite{eren2003sensor, servatius1999constraining} and later generalized to arbitrary dimensions together with formation controllers for undirected graphs in \cite{zhao2016bearing}, where infinitesimal bearing rigidity guarantees convergence to the target formation up to a global translation and scaling factor. In directed graphs, rigidity alone is insufficient, and additional constraint-consistency conditions give rise to the notion of bearing persistence \cite{zhao2015bearing}; related results on directed bearing rigidity are reported in \cite{eren2012formation}. For leader–first follower (LFF) formations, Trinh \emph{et al.}~\cite{trinh2018bearing} proposed bearing control laws that asymptotically stabilize the formation up to the leader’s translation and a common scale factor.

The scale invariance of bearing rigidity typically requires at least one inter-agent distance measurement to fix the scale and ensure convergence in both shape and size. This idea underpins, for example, Schiano \emph{et al.}~\cite{schiano2016rigidity}, who combined directed bearing rigidity in $\R^3 \times \S^1$ with a single distance measurement to fully determine the formation. More recently, Tang \emph{et al.}~\cite{tang2020bearing,tang2021formation} introduced the concepts of relaxed bearing rigidity and bearing persistently exciting (bearing-PE) formations, and showed that persistence of excitation (PE) of the desired bearings can replace classical bearing rigidity and constraint-consistency requirements while still guaranteeing uniform exponential stabilization of the formation shape and scale in leader–follower settings.

A key motivation for this work is the control of minimal-topology leader-follower bearing-PE formations, such as star graphs, in which each follower is connected to a single leader. Existing controllers provide exponential convergence rates that are explicitly related to the PE properties of the desired formation~\cite{tang2020bearing,tang2021formation}. Consequently, these rates can become arbitrarily small even for bearing-PE formations, limiting the performance achievable with fixed gains.

In this paper, we consider the problem of controlling a leader-follower formation (\textit{i.e.}, a formation under a directed acyclic graph that has a spanning tree, see Fig. \ref{fig:graph_topologies}) using only bearing and relative velocity measurements. The objective is to stabilize the formation to a desired geometric pattern whose time-varying bearings satisfy a PE condition.
We design a bearing-based formation controller that preserves the minimal sensing and relaxed-topology advantages of existing strategies, while decoupling the exponential convergence rate from the PE properties of the desired formation.
The proposed controller employs a time-varying matrix gain that adapts to the excitation accumulated along the bearing trajectories, which allows the convergence rate to be prescribed by the controller design parameters.
The control is further extended to double-integrator dynamics, and the theoretical results are validated through extensive simulations.

The paper is organized as follows. Section~\ref{sec:prelims} introduces the notation and the preliminaries on graph theory, persistence of excitation, and bearing-persistently exciting formations. Section~\ref{sec:1st_order_control} presents the proposed bearing-based leader-follower formation controller, establishes its stability and convergence properties for general sensing topologies. Section~\ref{sec:2nd_order_control} extends the design to double-integrator agent dynamics. Section~\ref{sec:simresults} illustrates the effectiveness of the proposed approach through simulation. 
Section~\ref{sec:conclusion} concludes the paper.

 \section{Preliminaries} \label{sec:prelims}

\subsection{Notation}
We denote by $\R$ and $\N$ the sets of reals and natural numbers, respectively.
The $3$-dimensional Euclidean space is denoted by $\R^3$. 
The Euclidean norm of a vector $\bm{x} \in \R^3$ is $\|\bm{x}\| = \sqrt{\bm{x}^\top \bm{x}}$.
For any $\bm{x} \in \mathbb{R}^3$, $\bm{x}^\times$ denotes the skew-symmetric matrix associated with the cross product, satisfying $\bm{x}^\times \bm{y} = \bm{x} \times \bm{y}$ for all $\bm{y}\in \R^3$.

We denote by $\mathbb{R}^{m \times n}$ the set of real $m \times n$ matrices, and by 
$I_n \in \mathbb{R}^{n \times n}$ the identity matrix. For a symmetric matrix 
$A \in \mathbb{R}^{n \times n}$, we write $A \succ 0$ (resp. $A \succeq 0$) to indicate 
that $A$ is positive definite (resp. positive semi-definite), and $A \preceq B$ to mean 
$B - A \succeq 0$.  
$\|A\|$ denotes the spectral norm of a matrix $A \in \R^{n \times n}$.
$\mathrm{blkdiag}(A_1, \dots, A_k)$ denotes the block-diagonal matrix with diagonal  blocks $A_1, \dots, A_k$. 

The unit sphere $\mathrm{S}^{2} := \{ \bm{y} \in \R^{3} \mid |\bm{y}| = 1 \} \subset \R^{3}$ denotes the set of unit vectors in $\R^3$. 
The projection onto the plane orthogonal to $\bm{y} \in \mathrm{S}^2$ is
$\pi_{\bm{y}} := I_3 - \bm{y}\bm{y}^\top = - \bm{y}^\times \bm{y}^\times \in \mathbb{R}^{3 \times 3}$, which satisfies $\pi_{\bm{y}} \succeq 0$, $\pi_{\bm{y}}^\top = \pi_{\bm{y}}$, and $\pi_{\bm{y}} \bm{y} = 0$.

\subsection{Bearing persistence of excitation}

The following definition and persistence of excitation properties are borrowed from \cite{le2017observers}.

\begin{definition} \label{def:PE}
Let $Q(t) \in \mathbb{R}^{n \times n}$ be a symmetric positive semi-definite matrix. 
It is said to be \emph{persistently exciting} (PE) if there exist constants 
$\delta, \mu > 0$ such that, for all $t \geq 0$,
\begin{equation} 
\frac{1}{\delta} \int_{t}^{t+\delta} Q(\tau)\, d\tau 
\;\succeq\; \mu I_n.
\label{eq:PE_matrix}
\end{equation}
\end{definition}

\begin{definition}\label{def:PE_y}
A time-varying direction $\bm{y}(t) \in \mathbb{S}^2$ is said to be \textbf{persistently exciting (PE)} if the associated projection $\pi_{\bm{y}(t)}$ satisfies the PE condition in Definition~1.
\end{definition}

\subsection{Preliminaries on graph theory}

Consider a system of $n$ $(n \geq 2)$ connected agents. The underlying interaction topology is modelled as a directed graph (digraph) $\mathcal{G} := (\mathcal{V}, \mathcal{E})$,
where $\mathcal{V} = \{1,2,\dots,n\}$ is the set of vertices and 
$\mathcal{E} \subseteq \mathcal{V} \times \mathcal{V}$ is the set of directed edges.
In this work, the graph is interpreted as a \emph{sensing graph}, meaning that if the ordered pair $(i,j) \in \mathcal{E}$, then agent $i$ can access or sense information about agent $j$, which is called a neighbor of agent $i$. 
The set of neighbors of agent $i$ is defined as $\mathcal{N}_i := \{\, j \in \mathcal{V} \mid (i,j) \in \mathcal{E} \,\}$.
Define $m_i = |\mathcal{N}_i|$, where $|\cdot|$ denotes the cardinality of a set.
A \emph{directed path} is a finite sequence of distinct vertices 
$\nu_1, \nu_2, \dots, \nu_k$ such that $(\nu_{i-1}, \nu_i) \in \mathcal{E}$, $2 \leq i \leq k$.


\begin{definition} \label{def:LF_formation}
A digraph $\mathcal{G} = (\mathcal{V}, \mathcal{E})$ has a \emph{leader-follower structure} 
if it is acyclic and has a directed spanning tree rooted at agent~$1$, \textit{i.e.}, 
$\mathcal{N}_1 = \emptyset$, $\mathcal{N}_2 = \{1\}$, and $\mathcal{N}_i \subseteq \{1, \dots, i-1\}$ for each $i \geq 3$. It has a \emph{minimal leader-follower structure} if each follower 
has exactly one neighbor. 
\end{definition}

This generalizes the leader-first follower (LFF) structure considered in \cite{trinh2018bearing}, where each follower has two neighbors, except for agent~$2$, which is connected only to the leader.

\begin{figure}[t]
    \centering
    \begin{subfigure}[t]{0.115\textwidth}
        \centering
        \includegraphics[width=.75\linewidth]{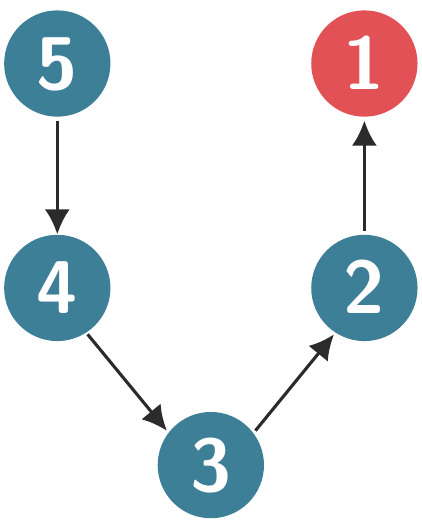}
        \caption{Chain}
    \end{subfigure}
    \begin{subfigure}[t]{0.115\textwidth}
        \centering
        \includegraphics[width=.98\linewidth]{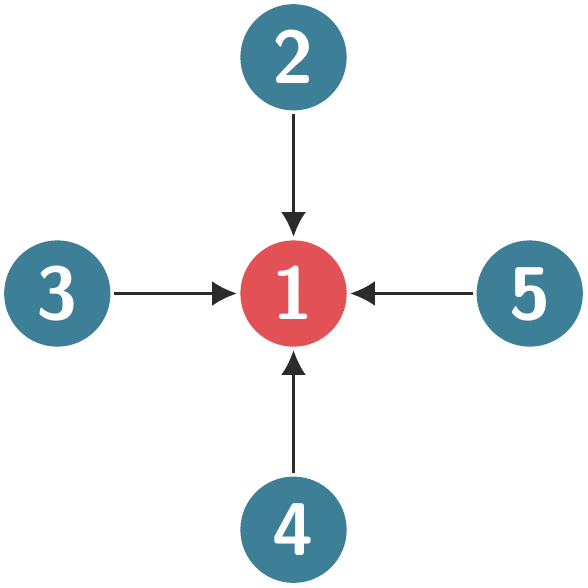}
        \caption{Star}
    \end{subfigure}
    \begin{subfigure}[t]{0.115\textwidth}
        \centering
        \includegraphics[width=.98\linewidth]{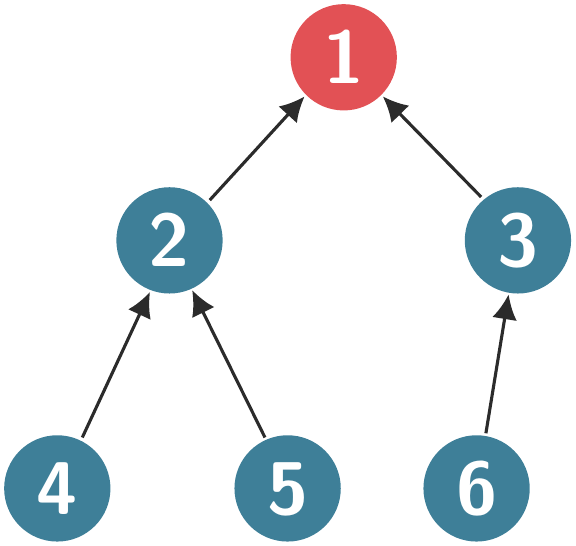}
        \caption{Tree}
    \end{subfigure}
    \begin{subfigure}[t]{0.115\textwidth}
        \centering
        \includegraphics[width=.98\linewidth]{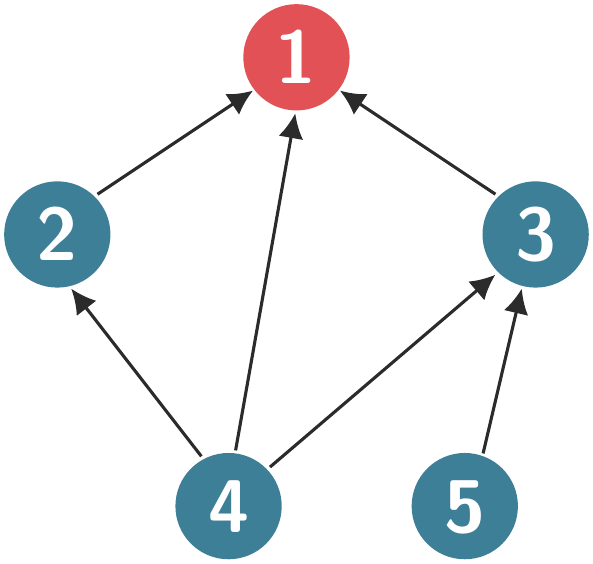}
        \caption{Mixed}
    \end{subfigure}
    \caption{Examples of leader-follower sensing digraph topologies. Agent~$1$ (red) 
    is the leader; agents $i \geq 2$ (blue) are followers. 
    The directed edge $(i,j) \in \mathcal{E}$ means that agent~$i$ senses information about agent~$j$.}
    \label{fig:graph_topologies}
\end{figure}

\section{Problem Description} \label{sec:problem_desc}

In this section, we formalize the bearing-based leader-follower formation control problem for multi-agent systems in three-dimensional space. 

Given a digraph $\mathcal{G}$, let $\bm{p}_i \in \mathbb{R}^3$, $i \in \mathcal{V}$, denote the position of each agent. The stacked vector $\bm{p} = \mathrm{col}(\bm{p}_1, \dots, \bm{p}_n) \in \mathbb{R}^{3n}$ is called a configuration of $\mathcal{G}$. The pair $(\mathcal{G},\bm{p})$ then defines a formation $\mathcal{G}(\bm{p})$ in 3D space.
Each agent has dynamics
\begin{equation} \label{eq:single_integrator}
    \dot{\bm{p}}_i = \bm{v}_i,
\end{equation}  
with $\bm{v}_i \in \R^3$ is the linear velocity of agent $i \in \mathcal{V}$.

For each follower $i \in \mathcal{V}\backslash \{1\}$, define the relative position with respect to a neighbor \(j\) as $\bm{p}_{ij} = \bm{p}_i - \bm{p}_j \in \R^3$. The corresponding bearing measurement is given by
\begin{equation} \label{eq:bearing_def}
    \bm{y}_{ij} = \frac{\bm{p}_{ij}}{\|\bm{p}_{ij}\|} \in \mathbb{S}^2.
\end{equation}
Similarly, define $\bm{v}_{ij} := \bm{v}_j - \bm{v}_i \in \R^3$ as the relative velocity between agent $i$ and $j$.

\begin{remark}
If agent $i$ has at least one bearing measurement $\bm{y}_{ij}$, $j\in \mathcal{N}_i$, that is PE according to Definition \ref{def:PE_y} or at least two non-collinear bearings, then the corresponding matrix $Q_i(t)$ is PE (see \cite[Lemma 1]{le2017observers} for details).
\end{remark}

\begin{definition} \label{def:BPE}
A formation $\mathcal{G}(\bm{p}(t))$ is called \textbf{bearing-PE} 
if, for all $i \in \mathcal{V}/\{1\}$, the matrix
\begin{align} \label{eq:BPE_condition}
    Q_i(t) := \sum_{j \in \mathcal{N}_i} \pi_{\bm{y}_{ij}(t)}
\end{align}
satisfies the PE condition in Definition~\ref{def:PE}.
\end{definition}

Consider a system of $n \geq 2$ agents interconnected through the leader-follower topology in Definition \ref{def:LF_formation}.
Define a time-varying desired formation $\mathcal{G}(\bm{p}^\star)$, where $\bm{p}^\star = \mathrm{col}(\bm{p}_1^\star, \dots, \bm{p}_n^\star) \in \R^{3n}$ and $\bm{p}_i^\star \in \R^{3}$ is the desired position of agent $i$.
The desired relative positions are denoted $\bm{p}_{ij}^\star$ and the corresponding desired bearings $\bm{y}_{ij}^\star$ are defined according to~\eqref{eq:bearing_def}.
The problem is to design a bearing-only control law that drives the relative position $\bm{p}_{ij}$ to their desired values $\bm{p}_{ij}^\star$ for all $i \in \mathcal{V}$.

\section{First-order bearing leader-follower formation control design} \label{sec:1st_order_control}

In this section, all agents are assumed to evolve according to the single integrator dynamics \eqref{eq:single_integrator}, with the velocity $\bm{v}_i$ taken as the control input.
Defining the relative position error as $\tilde{\bm{p}}_{ij} = \bm{p}_{ij} - \bm{p}_{ij}^\star \in \R^3$, $i, j \in \mathcal{V}$, its dynamics are
\begin{align} \label{eq:1st_dynamics}
    \dot{\tilde{\bm{p}}}_{ij} = \bm{v}_j - \bm{v}_j^\star - (\bm{v}_i - \bm{v}_i^\star),   
\end{align}
where $\bm{v}_i^\star$ and $\bm{v}_j^\star$ denote the desired velocities associated with $\mathcal{G}(\bm{p}^\star)$. 
The error $\tilde{\bm{p}}_{ij}$ is defined for all $i, j \in \mathcal{V}$, $i \neq j$, including pairs not belonging to $\mathcal{E}$. However, its dynamics are only considered for $(i, j) \in \mathcal{E}$.

\begin{assumption} \label{assump:BPE}
The desired formation $\mathcal{G}(\bm{p}^\star(t))$ is bearing-PE in the sense of Definition~\ref{def:BPE}, \textit{i.e.}, for each follower $i \in \mathcal{V}$,
$Q_i^\star(t) := \sum_{j \in \mathcal{N}_i} \pi_{\bm{y}_{ij}^\star(t)}$
satisfies the PE condition \eqref{eq:PE_matrix}.
\end{assumption} 

\begin{assumption} \label{assump:no_collision}
As the formation evolves in time, no inter-agent collisions or occlusions occur, such that the bearing measurements $\bm{y}_{ij}(t)$, for all $(i,j)\in\mathcal{E}$, remain well-defined for all $t \geq 0$.
\end{assumption}

The following result, taken from \cite[Theorem~2]{tang2020bearing}, establishes the exponential convergence of the follower errors under Assumptions~\ref{assump:BPE}-\ref{assump:no_collision}. 

\begin{proposition}
\label{prop:general_constant}
Consider system \eqref{eq:1st_dynamics} with the control law
\begin{align}\label{eq:control_general}
    \bm{v}_i = \bm{v}_i^\star - k_i \sum_{j \in \mathcal{N}_i} \pi_{\bm{y}_{ij}} \bm{p}_{ij}^\star,
    \qquad k_i > 0,
\end{align}
If Assumptions~\ref{assump:BPE}-\ref{assump:no_collision} hold,
then the equilibrium point $\tilde{\bm{p}}_{ij}(t) \to 0$ uniformly exponentially for all $i \in \mathcal{V}\backslash \{1\}$ and $j \in \mathcal{N}_i$. 
\end{proposition}

\begin{remark}
  The control law \eqref{eq:control_general} guarantees exponential convergence of the follower errors when the desired formation is bearing-PE. As is classical for time-varying systems stabilized through persistently exciting signals \cite{loria2002uniform}, the resulting exponential convergence rates in \cite{tang2020bearing,tang2021formation} depend on the control gains $k_i$ and on the parameters $\delta_i$ and $\mu_i$ characterizing the PE condition. 
  This explicit dependence can be restrictive since $\mu_i$ can be arbitrarily small even for a bearing-PE desired formation, and hence the guaranteed convergence rate can approach zero. This may occur, for example, for quasi-static minimal-topology desired formations whose bearings vary only marginally over the PE interval. 
\end{remark}

To decouple the convergence rate from the PE properties, we introduce a time-varying gain that evolves according to the Riccati-type equation \eqref{eq:Mi_riccati_general}, which accounts for the bearing excitation accumulated over time.
The following lemma establishes the well-posedness and uniform boundedness of this gain under the bearing-PE condition.

\begin{lemma}\label{lem:riccati_general}
Let $Q_i(t) := \sum_{j \in \mathcal{N}_i} \pi_{\bm{y}_{ij}(t)}$, and choose the matrix gain $M_i(t)$ as the solution to
\begin{equation}\label{eq:Mi_riccati_general}
    \dot{M}_i(t) = \lambda_i M_i(t) - M_i(t)\,Q_i(t)\,M_i(t), \quad M_i(0) \succ 0,
\end{equation}
with parameter $\lambda_i > 0$.
Suppose that $Q_i(t)$ is PE, \textit{i.e.}, there exist $\delta_i, \mu_i > 0$ such that $\frac{1}{\delta_i}\int_t^{t+\delta_i} Q_i(\tau)\,d\tau \succeq \mu_i I_3$, for all $t \geq 0$.
Then $M_i(t)$ is uniformly bounded, uniformly positive definite, and well-conditioned.
\end{lemma}

\begin{proof}
Uniform boundedness of $M_i(t)$ implies that there exist 
$\underline{m}_i, \overline{m}_i > 0$ such that for all $t \geq 0$, $\underline{m}_i\,I_3 \preceq M_i(t) \preceq \overline{m}_i\,I_3$.
To analyze the bounds of $M_i$, define $P_i := M_i^{-1}$. Differentiating and substituting \eqref{eq:Mi_riccati_general} gives
\begin{equation} \label{eq: P_ode}
    \dot{P}_i = -\lambda_i P_i + Q_i(t),
\end{equation} 
which yields the solution
\begin{equation} \label{eq:Pi_solution}
     P_i(t) = e^{-\lambda_i t}P_i(0) + \int_0^t e^{-\lambda_i(t-\tau)} Q_i(\tau)\,d\tau.
\end{equation}
This shows that $P_i(t)$ is at least semi-definite positive.

\noindent 
1) Upper bound: Since $Q_i(t) \preceq m_i I_3$ with $m_i = |\mathcal{N}_i|$ 
(as $\pi_{\bm{y}_{ij}} \preceq I_3$ for all $j \in \mathcal{N}_i$), we have
\begin{align*}
    P_i(t) &\preceq e^{-\lambda_i t}P_i(0) 
    + \frac{m_i}{\lambda_i}(1 - e^{-\lambda_i t}) \\
    &\preceq \max\left(\|P_i(0)\|, \frac{m_i}{\lambda_i}\right) I_3.
\end{align*}
\noindent
2) Lower bound: Using \eqref{eq:Pi_solution} over $[t-\delta_i, t]$:
\begin{align*}
    P_i(t) &= e^{-\lambda_i\delta_i}P_i(t-\delta_i) 
    + \int_{t-\delta_i}^{t} e^{-\lambda_i(t-\tau)} Q_i(\tau)\,d\tau \\
    &\succeq e^{-\lambda_i\delta_i}\int_{t-\delta_i}^{t} Q_i(\tau)\,d\tau
    \succeq e^{-\lambda_i\delta_i}\,\delta_i\mu_i\,I_3,
\end{align*}
where we used $P_i(t-\delta_i) \succeq 0$ and the PE condition on $Q_i$.
Thus $\frac{1}{\overline{m}_i} I_3 \preceq P_i(t) \preceq \frac{1}{\underline{m}_i} I_3$, with $\underline{m}_i = \lambda_i / m_i$ and $\overline{m}_i = e^{\lambda_i\delta_i}/(\delta_i\mu_i)$,
which implies that
\begin{equation} \label{eq:M_bounds}
    \frac{\lambda_i}{m_i}I_3 \preceq M_i(t) \preceq \frac{e^{\lambda_i\delta_i}}{\delta_i\mu_i}I_3, 
    \qquad \forall\, t \geq 0.
\end{equation}
Thus, $M_i(t)$ is uniformly lower and upper bounded and uniformly positive definite.
\end{proof}

\begin{remark}
An interpretation of the gain $M_i(t)$ follows from its inverse. It is clear that \eqref{eq: P_ode} corresponds to a first-order low-pass filter of the local bearing Gramian, with forgetting rate $\lambda_i$.  The PE condition on $Q_i(t)$ ensures that $P_i(t)$ remains uniformly positive definite, which in turn guarantees that $M_i(t):= P_i^{-1}(t)$ is well defined for all $t\geq0$.

\end{remark}

We now use the gain characterized in Lemma~\ref{lem:riccati_general} to define the control law \eqref{eq:control_riccati_general}. While the bounds on the gain depend on the bearing-PE properties, we show in the proof of Theorem~\ref{thm:riccati_general} that the convergence rate is prescribed by the design parameter $\lambda_i$, which can be independently selected.

\begin{theorem} \label{thm:riccati_general}
Suppose that Assumption \ref{assump:BPE} holds. Consider system \eqref{eq:1st_dynamics} with the control law
\begin{align}\label{eq:control_riccati_general}
    \bm{v}_i = \bm{v}_i^\star - M_i(t)\sum_{j \in \mathcal{N}_i}\pi_{\bm{y}_{ij}}\,\bm{p}_{ij}^\star,
\end{align}
where $M_i(t)$ is solution of \eqref{eq:Mi_riccati_general} with $Q_i(t) = \sum_{l\in\mathcal{N}_i}
\pi_{\bm{y}_{il}(t)}$. 
Then, the equilibrium point $\tilde{\bm{p}}_{ij} = 0$ is uniformly 
exponentially stable and $M_i$ is well-defined and uniformly bounded for all $i \in \mathcal{V}\backslash \{1\}$ and $j \in \mathcal{N}_i$.
\end{theorem}

\begin{proof} 
Consider the storage function $V_i = \frac{1}{2}\tilde{\bm{p}}_i^\top \Sigma_i \tilde{\bm{p}}_i$, where $\Sigma_i$ is a bounded positive semi-definite solution of 
\begin{equation} \label{eq:sigma_ode}
    \dot{\Sigma}_i = -\lambda_i \Sigma_i + Q_i, \quad \Sigma_i(0) := M_i^{-1}(0) \succ 0,
\end{equation}
which satisfies $\Sigma_i(t) M_i(t) = M_i(t) \Sigma_i(t) = I_3$ for all $t \geq 0$. We proceed by mathematical induction.

($i=2$): Agent 2 has $\mathcal{N}_2 = \{1\}$ with $\tilde{\bm{p}}_1=0$ and $\bm{v}_1 = \bm{v}_1^\star$.
Thus, under the control law in~\eqref{eq:control_riccati_general}, the error dynamics in~\eqref{eq:1st_dynamics} reduce to
$ \dot{\tilde{\bm p}}_{21} = -M_2\pi_{\bm y_{21}}\tilde{\bm p}_{21}.$
First, notice that $\tilde{\bm p}_{21}$ is bounded. Indeed,
using \eqref{eq:sigma_ode} and $\Sigma_2M_2=I_3$ yields
\begin{align*}
    \frac{d}{dt}
    \big(\Sigma_2\tilde{\bm p}_{21}\big)
    &= - \lambda_2\Sigma_2\tilde{\bm p}_{21} + \pi_{\bm y_{21}} \tilde{\bm p}_{21} - \Sigma_2 M_2\pi_{\bm y_{21}}\tilde{\bm p}_{21} \\
    &=
    -\lambda_2\Sigma_2\tilde{\bm p}_{21}.
\end{align*}
Since $\Sigma_2(t)= e^{-\lambda_2t} \left( \Sigma_2(0) + \int_0^t e^{\lambda_2\tau} \pi_{\bm y_{21}(\tau)}\,d\tau \right)$,
we have
\[
    \tilde{\bm p}_{21}(t)
    =
    \left(
    \Sigma_2(0)
    +
    \int_0^t e^{\lambda_2\tau}
    \pi_{\bm y_{21}(\tau)}\,d\tau
    \right)^{-1}
    \Sigma_2(0)\tilde{\bm p}_{21}(0),
\]
which is bounded since the matrix in parentheses is lower bounded by
$\Sigma_2(0)\succ0$.

Differentiating $V_2$ and substituting \eqref{eq:sigma_ode} yields
\begin{align}
    \dot{V}_2
    &= \tfrac{1}{2} \dot{\tilde{\bm{p}}}_{21}^\top
       \Sigma_2\tilde{\bm{p}}_{21}
     + \tfrac{1}{2}\tilde{\bm{p}}_{21}^\top
       \Sigma_2\dot{\tilde{\bm{p}}}_{21}
     + \tfrac{1}{2}\tilde{\bm{p}}_{21}^\top
       \dot{\Sigma}_2\tilde{\bm{p}}_{21} \notag\\
    &= -\tfrac{\lambda_2}{2}
       \tilde{\bm{p}}_{21}^\top\Sigma_2\tilde{\bm{p}}_{21}
       -\tfrac{1}{2}
       \tilde{\bm{p}}_{21}^\top
       \pi_{\bm{y}_{21}}\tilde{\bm{p}}_{21}\notag\\
    &\leq
       -\tfrac{1}{2}
       \tilde{\bm{p}}_{21}^\top
       \pi_{\bm{y}_{21}}\tilde{\bm{p}}_{21}.
    \label{eq:Vdot_M_1st}
\end{align}
Since $\pi_{\bm y_{21}}\bm p_{21}=0$, we obtain
$\dot V_2
    \leq
    -\frac{1}{2}
    \bm p_{21}^{\star\top}
    \pi_{\bm y_{21}}\bm p_{21}^\star
    =
    -\frac{1}{2}\|\bm p_{21}^\star\|^2
    \bm y_{21}^{\star\top}
    \pi_{\bm y_{21}}\bm y_{21}^\star$.
Moreover, using
\begin{equation}\label{eq:y_ystar_identity}
    \bm y_{21}^{\star\top}\pi_{\bm y_{21}}\bm y_{21}^\star
    =
    \|\bm y_{21}^\star\times\bm y_{21}\|^2
    =
    \bm y_{21}^{\top}\pi_{\bm y_{21}^\star}\bm y_{21},
\end{equation}
it follows that
\[
    \dot V_2
    \leq
    -\frac{1}{2}
    \frac{\|\bm p_{21}^\star\|^2}{\|\bm p_{21}\|^2}
    \tilde{\bm p}_{21}^\top
    \pi_{\bm y_{21}^\star}
    \tilde{\bm p}_{21}.
\]
Since $\tilde{\bm p}_{21}$ is bounded,
$\bm p_{21}=\bm p_{21}^\star+\tilde{\bm p}_{21}$ is bounded.
Assuming that $\bm p_{21}^\star$ is bounded away from zero, define
\begin{equation}\label{eq:gamma2}
    \gamma_2
    := \frac{\inf_{t\geq0}\|\bm p_{21}^\star(t)\|^2}{\sup_{t\geq0}\|\bm p_{21}(t)\|^2 } >0.
\end{equation}
Hence,
$\dot V_2
    \leq
    -\frac{\gamma_2}{2}
    \tilde{\bm p}_{21}^\top
    \pi_{\bm y_{21}^\star}
    \tilde{\bm p}_{21}
    \leq 0$.
By Assumption~\ref{assump:BPE} and \cite[Lemma~5]{loria2002uniform},
it follows that $\tilde{\bm p}_{21}\to0$ uniformly exponentially.
Lemma~\ref{lem:piy_PE} then implies that there exists $t_0>0$
such that $\bm y_{21}$ is PE for all $t\geq t_0$.
Thus, by Lemma~\ref{lem:riccati_general}, $M_2(t)$ is uniformly
bounded and uniformly positive definite for all $t\geq t_0$.
Then, $\frac{1}{\overline{m}_2}I_3 \preceq M_2^{-1}(t) = \Sigma_2(t)
    \preceq \frac{1}{\underline{m}_2}I_3$.
Hence, $V_2$ is uniformly positive definite and radially unbounded in $\tilde{\bm{p}}_{21}$, and therefore a Lyapunov function.
Moreover, from the expression of $\dot{V}_2$, we have
$\dot{V}_2 \leq -\frac{\lambda_2}{2}\tilde{\bm{p}}_{21}^\top \Sigma_2 \tilde{\bm{p}}_{21} = -\lambda_2\,V_2$, which yields $V_2(t) \leq e^{-\lambda_2 (t-t_0)} V_2(t_0)$ for all $t \geq t_0$. Finally, since 
$\frac{1}{\overline{m}_2} \|\tilde{\bm{p}}_{21}\|^2 \leq V_2 \leq \frac{1}{\underline{m}_2} \|\tilde{\bm{p}}_{21}\|^2$, we conclude that 
$\|\tilde{\bm{p}}_{21}(t)\| \leq \sqrt{\frac{\overline{m}_2}{\underline{m}_2}}e^{-\lambda_2 (t-t_0)/2} \|\tilde{\bm{p}}_{21}(t_0)\|$.

\noindent
($i=3$):
Assume that $\tilde{\bm{p}}_{21}(t)\to 0$ uniformly exponentially from the case $i=2$. For agent $3$, we have $\mathcal{N}_3 \subseteq \{1,2\}$. Using \eqref{eq:2nd_err_dynamics} and \eqref{eq:2nd_control_law}, the closed-loop system for $\tilde{\bm{p}}_{3j}$ ($j \in \mathcal{N}_3$) can be written as
\begin{equation}\label{eq:dot_p3j_full}
\dot{\tilde{\bm{p}}}_{3j}
= - M_3 Q_3\tilde{\bm{p}}_{3j} + \bm{d}_{3j}(t),
\end{equation} 
with $Q_3 = \sum_{l\in\mathcal{N}_3}\pi_{\bm{y}_{3l}}$ and $\bm{d}_{3j}(t)$ a bounded function that, depending on the graph, satisfies
\begin{itemize}
    \item[$\bullet$] $\mathcal{N}_3 = \{1\}$: $\bm{d}_{31} = 0$, $\dot{\tilde{\bm{p}}}_{31}
= - M_3\pi_{\bm{y}_{31}}\tilde{\bm{p}}_{31}$, and the result follows directly from the case $i=2$. 
    \item[$\bullet$] $\mathcal{N}_3 = \{2\}$: $\bm{d}_{31} = M_2\pi_{\bm{y}_{21}} \tilde{\bm{p}}_{21} \to 0$.
    \item[$\bullet$] $\mathcal{N}_3 = \{1,2\}$: $\bm{d}_{31} = - M_3\pi_{\bm{y}_{32}} \tilde{\bm{p}}_{21}$, and $\bm{d}_{32} = (M_2\pi_{\bm{y}_{21}} - M_3\pi_{\bm{y}_{31}}) \tilde{\bm{p}}_{21} \to 0$.
\end{itemize}
Therefore, system \eqref{eq:dot_p3j_full} can be interpreted as a cascaded system that has $\tilde{\bm{p}}_{21}$ as input to the unforced system
\begin{equation}\label{eq:dot_p3j}
\dot{\tilde{\bm{p}}}_{3j}
= - M_3(t) Q_3(t) \tilde{\bm{p}}_{3j}, \qquad j \in \mathcal{N}_3.
\end{equation}
Consider the storage function 
$V_3=\frac{1}{2}\tilde{\bm{p}}_{3j}^{\top}\Sigma_3\tilde{\bm{p}}_{3j}$ for the unforced system \eqref{eq:dot_p3j}.
Its time derivative verifies 
\begin{equation*}
\begin{aligned}
\dot{V}_3
&= \tfrac{1}{2} \dot{\tilde{\bm{p}}}_{3j}^\top \Sigma_3\tilde{\bm{p}}_{3j} + \tfrac{1}{2}\tilde{\bm{p}}_{3j}^\top \Sigma_3\dot{\tilde{\bm{p}}}_{3j} 
     + \tfrac{1}{2}\tilde{\bm{p}}_{3j}^\top \dot{\Sigma}_3\tilde{\bm{p}}_{3j} \notag \\
&= -\tfrac{\lambda_3}{2}\tilde{\bm{p}}_{3j}^{\top}\Sigma_3\tilde{\bm{p}}_{3j} -\tfrac{1}{2}\tilde{\bm{p}}_{3j}^{\top}Q_3 \tilde{\bm{p}}_{3j} \leq -\tfrac{1}{2}\tilde{\bm{p}}_{3j}^{\top}Q_3\tilde{\bm{p}}_{3j}.
\end{aligned}
\end{equation*}
Therefore, $\dot{V}_3\leq 0$, and the unforced system is bounded. 
Since $\tilde{\bm{p}}_{3j} = \tilde{\bm{p}}_{3l} + \tilde{\bm{p}}_{lj}$, $l \neq j$, $j \in \{1, 2\}$ and $\tilde{\bm{p}}_{21} = 0$ in the unforced system \eqref{eq:dot_p3j}, it follows that $\tilde{\bm{p}}_{3l} = \tilde{\bm{p}}_{3j}$.
Using the identity
$\bm{y}_{3l}^{\star\top}\pi_{\bm{y}_{3l}}\bm{y}_{3l}^\star =\bm{y}_{3l}^{\top}\pi_{\bm{y}_{3l}^\star}\bm{y}_{3l}$, we obtain
\begin{align*}
\dot{V}_3 \leq - \tfrac{1}{2}
\tilde{\bm{p}}_{3j}^{\top}\sum_{l\in\mathcal{N}_3}\tfrac{\|\bm{p}_{3l}^{\star}\|^2}{\|\bm{p}_{3l}\|^2}\pi_{\bm{y}_{3l}^\star}\tilde{\bm{p}}_{3j}
\leq - \tfrac{\gamma_3}{2}
\tilde{\bm{p}}_{3j}^{\top}
Q_3^\star
\tilde{\bm{p}}_{3j},
\end{align*}
where $\gamma_3 := \min_{l\in\mathcal{N}_3}
    \frac{\inf_{t\geq0}
        \|\bm p_{3l}^\star(t)\|^2 }{ \sup_{t\geq0}
        \|\bm p_{3l}(t)\|^2 } >0.$
Using the same argument as in Case $i=2$, the equilibrium $\tilde{\bm{p}}_{3j}=0$ of the unforced system is uniformly exponentially stable, and $\|\tilde{\bm{p}}_{3j}(t)\| \leq \sqrt{\frac{\overline{m}_3}{\underline{m}_3}}e^{-\lambda_3 (t-t_0)/2} \|\tilde{\bm{p}}_{3j}(t_0)\|$.
Finally, since the full system \eqref{eq:dot_p3j_full} is a cascade driven by $\tilde{\bm{p}}_{21}$, which is already exponentially convergent to zero, it implies that $\tilde{\bm{p}}_{3j}\to 0$ uniformly exponentially.

\noindent
($i\geq 4$):
Assume that the claim holds for all $2\leq k\leq i-1$, that is,
$\tilde{\bm{p}}_{kj}(t)\to 0$ uniformly exponentially for all $j\in\mathcal{N}_k$ and all $2\leq k\leq i-1$.
Recall that, under the control law in~\eqref{eq:control_riccati_general}, the closed-loop system for the states
$\tilde{\bm{p}}_{ij}$, $j\in\mathcal{N}_i$, can be written as
\begin{align}
\dot{\tilde{\bm{p}}}_{ij}
= -M_i \pi_{\bm{y}_{ij}}\tilde{\bm{p}}_{ij} -M_i\sum_{ l\in\mathcal{N}_i\backslash\{j\}}\pi_{\bm{y}_{il}}\tilde{\bm{p}}_{il}
+\bm{v}_j-\bm{v}_j^\star,
\label{eq:gen_i_closedloop}
\end{align}
where $\bm{v}_j$ is a function of the error variables $\tilde{\bm{p}}_{km}$, $2\leq k\leq i-1$, $m\in\mathcal{N}_k$, $\tilde{p}_{il} = \tilde{p}_{ij} + \tilde{p}_{jl}$. Note that since the graph is connected, $\tilde{p}_{jl}$ can be represented by the error variables $\tilde{\bm{p}}_{km}$, $2\leq k\leq i-1$, $m\in\mathcal{N}_k$. Hence, the dynamics~\eqref{eq:gen_i_closedloop} can then be considered as a cascaded
system with inputs $\tilde{\bm{p}}_{km}$, $2\leq k\leq i-1$, $m\in\mathcal{N}_k$, of the unforced system
\begin{align} \dot{\tilde{\bm{p}}}_{ij}
=-M_i Q_i \tilde{\bm{p}}_{ij},
\qquad j\in\mathcal{N}_i.
\label{eq:gen_i_unforced}
\end{align}
Consider the storage function $V_i=\frac{1}{2}\tilde{\bm{p}}_{ij}^{\top}M_i^{-1}\tilde{\bm{p}}_{ij}$.
Its derivative along~\eqref{eq:gen_i_unforced} is
\begin{align}
\dot{V}_i
&=
-\tfrac{\lambda_i}{2}\tilde{\bm{p}}_{ij}^{\top}\Sigma_i \tilde{\bm{p}}_{ij}
-\tfrac{1}{2}\tilde{\bm{p}}_{ij}^{\top}Q_i \tilde{\bm{p}}_{ij} \leq -\tfrac{1}{2}\tilde{\bm{p}}_{ij}^{\top}Q_i \tilde{\bm{p}}_{ij}.
\label{eq:Vdot_i_general}
\end{align}
Therefore, $\dot{V}_i\leq 0$. Using similar arguments to the case $i=3$, $\tilde{\bm{p}}_{ij} = \tilde{\bm{p}}_{il} + \tilde{\bm{p}}_{lj}$, $l \in \mathcal{N}_i\backslash\{j\}$ and $\tilde{\bm{p}}_{lj} = 0$ in the unforced system \eqref{eq:gen_i_unforced}, it follows that $\tilde{\bm{p}}_{ij} = \tilde{\bm{p}}_{il}$.
Then, using the identity
$\bm{y}_{il}^{\star\top}\pi_{\bm{y}_{il}}\bm{y}_{il}^\star
=\bm{y}_{il}^{\top}\pi_{\bm{y}_{il}^\star}\bm{y}_{il}$, 
we obtain
\begin{align*}
\dot{V}_i
\leq - \tfrac{1}{2}
\tilde{\bm{p}}_{ij}^{\top}\sum_{l\in\mathcal{N}_i}\tfrac{\|\bm{p}_{il}^{\star}\|^2}{\|\bm{p}_{il}\|^2}\pi_{\bm{y}_{il}^\star}\tilde{\bm{p}}_{ij}
\leq -\tfrac{\gamma_i}{2}
\tilde{\bm{p}}_{ij}^{\top}
Q^\star_i
\tilde{\bm{p}}_{ij},
\end{align*}
where $\gamma_i := \min_{l\in\mathcal{N}_i} \frac{\inf_{t\geq0}\|\bm{p}_{il}^\star(t)\|^2}{\sup_{t\geq0}\|\bm{p}_{il}(t)\|^2} >0.$
From there, we conclude that the equilibrium $\tilde{\bm{p}}_{ij}=0$ of the unforced system~\eqref{eq:gen_i_unforced}
is uniformly exponentially stable, and 
\begin{equation} \label{eq:convergence_rate_1st}
    \|\tilde{\bm{p}}_{ij}(t)\| \leq \sqrt{\frac{\overline{m}_i}{\underline{m}_i}}e^{-\lambda_i (t-t_0)/2} \|\tilde{\bm{p}}_{ij}(t_0)\|
\end{equation} 
for all $i\in \mathcal{V}\backslash\{1\}$ and $j \in \mathcal{N}_i$.
Since $\bm{v}_j-\bm{v}_j^\star$ is uniformly exponentially vanishing, the cascaded system~\eqref{eq:gen_i_closedloop} is also uniformly exponentially stable.
Hence, $\tilde{\bm{p}}_{ij}(t)\to 0$ uniformly exponentially for all $i\in \mathcal{V}\backslash\{1\}$ and $j\in\mathcal{N}_i$.
\end{proof}

The bound in \eqref{eq:convergence_rate_1st} shows that, unlike for the fixed gain design \eqref{eq:control_general}, the exponential convergence rate for agent $i$ is determined by the parameter $\lambda_i$ involved in \eqref{eq:Mi_riccati_general}, whereas $\delta_i$ and $\mu_i$ characterize the upper bound on $M_i(t)$ given in~\eqref{eq:M_bounds}, and are hence reflected in the constant $\sqrt{\overline m_i/\underline m_i}$.

\section{Second-order double integrator dynamics} \label{sec:2nd_order_control}

The design presented in Section \ref{sec:1st_order_control} assumes that each agent can 
directly control its velocity. In practice, many robotic platforms have the control input appearing at the acceleration level. To address this more realistic setting, we extend the proposed design to agents governed by double-integrator motion dynamics.
Consider the formation $\mathcal{G}(\bm{p})$, where each agent $i \in \mathcal{V}$ evolves according to:
\begin{align} \label{eq:2nd_dynamics}
    \dot{\bm{p}}_i &= \bm{v}_i, &
    \dot{\bm{v}}_i &= \bm{u}_i, 
\end{align}
where $u_i \in \R^3$ is the acceleration control input.
Consider a desired trajectory $(\bm{p}_{i}^\star(t),\bm{v}_{i}^\star(t),\bm{u}_i^\star(t))$ satisfying $\dot{\bm{p}}_{i}^\star = \bm{v}_{i}^\star$, $\dot{\bm{v}}_{i}^\star = \bm{u}_i^\star$.
Define the relative position and velocity errors as
\begin{align}  \label{eq:2nd_errors}
    \tilde{\bm{p}}_{ij} &:= \bm{p}_{ij} - \bm{p}_{ij}^\star,\qquad
    \tilde{\bm{v}}_{ij} := \bm{v}_{ij} - \bm{v}_{ij}^\star.
\end{align}
From \eqref{eq:2nd_dynamics} and \eqref{eq:2nd_errors}, the error dynamics satisfy
\begin{equation}
\begin{cases}
    \dot{\tilde{\bm{p}}}_{ij} &= \tilde{\bm{v}}_{ij}, \\
    \dot{\tilde{\bm{v}}}_{ij} &= \bm{u}_j - \bm{u}_j^\star - (\bm{u}_i - \bm{u}_i^\star).   
\end{cases}
    \label{eq:2nd_err_dynamics} 
\end{equation}

\begin{assumption} \label{assump:2}
The desired acceleration $\bm{u}^\star_i(t)$ and the desired
relative velocity $\bm{v}^\star_{ij}(t)$ are bounded for all $t \geq 0$.
\end{assumption}

For the double-integrator dynamics, the following result, adapted from \cite[Theorem~2]{tang2021formation}, establishes exponential convergence of the follower errors.

\begin{proposition}
\label{prop:general_constant}
Consider system \eqref{eq:2nd_dynamics} with the control law
\begin{align}\label{eq:control_general_2nd}
    \bm{u}_i = \bm{u}_i^\star - k_{p_i} \sum_{j \in \mathcal{N}_i} \pi_{\bm{y}_{ij}} \bm{p}_{ij}^\star + k_{v_i} \sum_{j \in \mathcal{N}_i} \tilde{\bm{v}}_{ij},
\end{align}
where $k_{v_i}$ and $k_{p_i}$ are positive gains that satisfy $k_{v_i} > \frac{1}{m_i}$ and $k_{p_i} < \frac{4}{m_i} - \frac{4}{k_{vi}^2 m_i^3}$, $m_i = |\mathcal{N}_i|$.
If Assumptions \ref{assump:BPE}–\ref{assump:2} are satisfied, then the equilibrium $(\tilde{\bm{p}}_{ij},\tilde{\bm{v}}_{ij}) = (0,0)$ is uniformly exponentially stable, $\forall i \in \mathcal{V} \backslash \{1\}$ and $\forall j \in \mathcal{N}_i$. 
\end{proposition}

As in the first-order case, Proposition~\ref{prop:general_constant} ensures exponential stability, but the corresponding convergence rates remain dependant on the desired formation's bearing-PE properties.
We next establish the convergence result for system \eqref{eq:2nd_dynamics} using the proposed gain. The resulting controller yields exponential convergence with rates set by the parameters $\lambda_i$.

\begin{theorem}
Consider system \eqref{eq:2nd_dynamics} with the control law for each agent $i\in \mathcal{V}$
\begin{align}
    \bm{u}_i 
    &= \bm{u}_i^\star -  M_i \sum_{j \in \mathcal{N}_i} \pi_{\bm{y}_{ij}}\bm{p}_{ij}^\star + k_{v_i} \sum_{j \in \mathcal{N}_i}\tilde{\bm{v}}_{ij}, \label{eq:2nd_control_law}
\end{align}
where $M_i(t) \in \mathbb{R}^{3\times 3}$ are solutions of \eqref{eq:Mi_riccati_general} for all $i \in \mathcal{V}\backslash\{1\}$ and $k_{v_i}$ are constant positive gains.
Suppose that Assumptions \ref{assump:BPE}–\ref{assump:2} are satisfied.
Then, $M_i$ is uniformly bounded and uniformly positive definite, and
the equilibrium $(\tilde{\bm{p}}_{ij},\tilde{\bm{v}}_{ij}) = (0,0)$ of the closed-loop error system is uniformly exponentially stable, $\forall i \in \mathcal{V} \backslash \{1\}$ and $\forall j \in \mathcal{N}_i$.
\end{theorem}

\begin{proof}
Define the state $\tilde{\bm{x}}_{ij} := (\tilde{\bm{p}}_{ij},
        \tilde{\bm{v}}_{ij})
    \in \mathbb{R}^6$.
    Consider the storage function $V_i(\tilde{\bm{x}}_{ij}) = \tilde{\bm{x}}_{ij}^\top P_i \tilde{\bm{x}}_{ij}$,
where
\begin{align}
    P_i &:= 
    \begin{bmatrix}
        \Sigma_i & \beta \Sigma_i \\
        \beta \Sigma_i & \Sigma_i
    \end{bmatrix}, \;\; 0<\beta<1, \;\; i \in \mathcal{V}\backslash\{1\}.
\end{align}
where $\Sigma_i$ is a bounded semi-definite positive solution of \eqref{eq:sigma_ode}.
We proceed by mathematical induction. 

\noindent
($i=2$): Agent 2 has $\mathcal{N}_2 = \{1\}$ and $(\tilde{\bm{p}}_1, \tilde{\bm{v}}_1)=(0,0)$.
From \eqref{eq:2nd_control_law}, the error dynamics \eqref{eq:2nd_err_dynamics} satisfy
\begin{equation}
\begin{cases}
    \dot{\tilde{\bm{p}}}_{21} &= \tilde{\bm{v}}_{21}, \\
    \dot{\tilde{\bm{v}}}_{21} &= - M_2  \pi_{\bm{y}_{21}} \tilde{\bm{p}}_{21} - k_{v_2} \tilde{\bm{v}}_{21}. 
\end{cases}
    \label{eq:2nd_err_dynamics_2} 
\end{equation}
Then, the closed-loop error dynamics \eqref{eq:2nd_err_dynamics} are
\begin{align} \label{eq:xtilde_dynamics}
    \dot{\tilde{\bm{x}}}_{21} = -A_2(t)\tilde{\bm{x}}_{21}, \;  A_2(t) :=
    \begin{bmatrix}
        0 & -I_3 \\
        M_2(t)Q_2(t) & k_{v_2} I_3
    \end{bmatrix}.
\end{align}
with $Q_2(t) = \pi_{\bm{y}_{21}(t)}$.

Differentiating $V_2$ along the trajectories of \eqref{eq:xtilde_dynamics}, we obtain
\begin{align}
    \dot{V}_2 
    &= \tilde{\bm{x}}_{21}^\top P_2 \dot{\tilde{\bm{x}}}_{21} 
     + \dot{\tilde{\bm{x}}}_{21}^\top P_2 \tilde{\bm{x}}_{21} + \tilde{\bm{x}}_{21}^\top \dot{P}_2 \tilde{\bm{x}}_{21}  \notag\\
    &= \tilde{\bm{x}}_{21}^\top \big(P_2 A_2(t) + A_2^\top(t) P_2 + \dot{P}_2 \big)\tilde{\bm{x}}_{21}. \label{eq:2nd_lyap_derivative}
\end{align}
Denote $S_2(t) := - \big(P_2 A_2(t) + A_2^\top(t) P_2 + \dot{P}_2\big)$.
Then \eqref{eq:2nd_lyap_derivative} becomes
$\dot{V}_2 = - \tilde{\bm{x}}_{21}^\top S_2(t)\tilde{\bm{x}}_{21}$, with $S_2(t) = S_\pi + S_\Sigma$,
where 
\begin{align*}
    S_\pi &= \begin{bmatrix}
       (2 \beta - 1)\pi_{\bm{y}_{21}} & (1 - \beta) \pi_{\bm{y}_{21}}  \\
        (1 - \beta) \pi_{\bm{y}_{21}} &  - \pi_{\bm{y}_{21}} 
    \end{bmatrix},  \\
    S_\Sigma &= \begin{bmatrix}
       \lambda_2 \Sigma_2  & (\beta (\lambda_2 + k_{v_2}) - 1) \Sigma_2 \\
        (\beta (\lambda_2 + k_{v_2}) - 1) \Sigma_2 & (\lambda_2 + 2 k_{v_2} - 2\beta ) \Sigma_2
    \end{bmatrix} .
\end{align*}
Note that, by choosing $\beta = \frac{1}{\lambda_2 + k_{v_2}}$, we get
\begin{align*}
    S_\Sigma &= \begin{bmatrix}
       \lambda_2 \Sigma_2  & 0\\
        0 & (\lambda_2 + 2 k_{v_2} - \frac{2}{\lambda_2 + k_{v_2}} ) \Sigma_2  
    \end{bmatrix} .
\end{align*}

To show that $M_2$ is uniformly bounded and uniformly positive definite, we prove that $\pi_{\bm{y}_{21}}$ is PE provided that $\pi_{\bm{y}_{21}^\star}$ is PE. Suppose, by contradiction, that $\pi_{\bm{y}_{21}}$ is not PE. 
Then, using Lemma \ref{lem:nonPE_condition},
one can consider $\bm{y}_{21}(t) \to \bar{\bm{y}}_{21}$ constant.

From \eqref{eq:sigma_ode}, the lack of PE of $\pi_{\bm{y}_{21}}$ implies that the Riccati solution satisfies $\Sigma_2 \to \frac{1}{\lambda_2}\pi_{\bar{\bm{y}}_{21}}$. 
Thus $M_2 = \Sigma_2^{-1}$ becomes unbounded outside $\mathrm{Im}(\pi_{\bar{\bm{y}}_{21}})$. 
By Lemma~\ref{lem:projected_inverse} applied to $\Sigma_2$, the projected pseudo-inverse satisfies
$ \pi_{\bar{\bm{y}}_{21}} M_2 \pi_{\bar{\bm{y}}_{21}} = \lambda_2 \pi_{\bar{\bm{y}}_{21}}$ on $\mathrm{Im}(\pi_{\bar{\bm{y}}_{21}}) $, and given that $\pi_{\bm{y}_{21}} \to \pi_{\bar{\bm{y}}_{21}}$, we have $\pi_{\bm{y}_{21}} M_2 \pi_{\bm{y}_{21}} \rightarrow \lambda_2 \pi_{\bm{y}_{21}}$ on $\mathrm{Im}(\pi_{\bm{y}_{21}})$.

Define the projected state $\bar{\bm{x}}_{21} = \Pi_{2}\tilde{\bm{x}}_{21}$ and $\bar{P}_2 = \Pi_{2} P_2 \Pi_{2}$, where $\Pi_{2} = \mathrm{blkdiag}(\pi_{\bm{y}_{21}}, \pi_{\bm{y}_{21}})$, and consider the storage function $\bar{V}_2 = \bar{\bm{x}}_{21}^\top \bar{P}_2 \bar{\bm{x}}_{21} $.  
Using the fact that $\Sigma_2 \to \frac{1}{\lambda_2}\pi_{\bm{y}_{21}}$ and $\Pi_{2}\bar{\bm{x}}_{21} = \bar{\bm{x}}_{21}$, one verifies by direct computation that
\begin{align*}
    \dot{\bar{V}}_2 
    &= - \bar{\bm{x}}_{21}^\top 
    \begin{bmatrix}
        2\beta I_3 & (1-\beta) I_3 \\ 
        (1-\beta) I_3 & \frac{2}{\lambda_2}(k_{v_2} - \beta) I_3
    \end{bmatrix}
    \bar{\bm{x}}_{21}.
\end{align*}
Hence, if the gains satisfy $4\beta(k_{v_2} - \beta) > \lambda_2 (1-\beta)^2$, the above matrix is positive definite on $\mathrm{Im}(\Pi_2)$, which implies $\Pi_2 \tilde{\bm{x}}_{21} \to 0$.
That is, $\pi_{\bm{y}_{21}} \tilde{\bm{p}}_{21} = \pi_{\bm{y}_{21}} \tilde{\bm{v}}_{21} = 0$.

From \eqref{eq:2nd_control_law}, we have 
$\ddot{\bm{p}}_{21} = \bm{u}_{1}^\star - \bm{u}_{2}^\star + M_2 \pi_{\bm{y}_{21}}\bm{p}_{21}^\star - k_{v_2}\tilde{\bm{v}}_{21}$.
Multiplying on the left by $\pi_{\bm{y}_{21}}$ and exploiting the fact that $\bm{u}_1^\star - \bm{u}_2^\star = \ddot{\bm{p}}_{21}^\star$, and that $\bm{y}_{21}$ lies in the same direction as $\bm{p}_{21}$ and $\bm{v}_{21}$ and consequently as $\ddot{\bm{p}}_{21}$, one obtains 
$$
    \pi_{\bm{y}_{21}}( \ddot{\bm{p}}_{21}^\star + k_{v_2}\dot{\bm{p}}_{21}^\star + M_2 \pi_{\bm{y}_{21}}\bm{p}_{21}^\star) = 0.
$$
or equivalently:
$$
  \pi_{\bm{y}_{21}}( \ddot{\bm{p}}_{21}^\star + k_{v_2}\dot{\bm{p}}_{21}^\star + \lambda_2 \bm{p}_{21}^\star) = 0.
$$
From there, one concludes that the desired trajectory of $\bm{p}_{21}^\star$, solution of the above differential equation,  is in the direction of $\bm{p}_{21}$. Hence $\bm{y}_{21}^\star$ is not PE, which contradicts the assumption, which in turn implies that $\bm{y}_{21}$ is PE. 

Therefore, by the PE property of $\pi_{\bm{y}_{21}}$ and Lemma \ref{lem:riccati_general}, $M_2$ is well-defined and uniformly positive definite, and $\Sigma_2 = M_2^{-1}$ satisfies uniform bounds $\frac{1}{\overline{m}_2} I_3 \preceq \Sigma_2 \preceq \frac{1}{\underline{m}_2} I_3$. Consequently, $V_2$ is positive definite and radially unbounded in $\tilde{\bm{x}}_{21}$.

Now, considering the expression \eqref{eq:2nd_lyap_derivative}, we have 
\begin{align*}
\tilde{\bm{x}}_{21}^\top S_\pi \tilde{\bm{x}}_{21} &= (2\beta - 1)\tilde{\bm{p}}_{21}^\top \pi_{\bm{y}_{21}} \tilde{\bm{p}}_{21}
+ 2(1-\beta)\tilde{\bm{p}}_{21}^\top \pi_{\bm{y}_{21}} \tilde{\bm{v}}_{21} \\ & \hspace{4.5cm}
- \tilde{\bm{v}}_{21}^\top \pi_{\bm{y}_{21}} \tilde{\bm{v}}_{21}.
\end{align*}
Using Young's inequality,
\begin{align*} 
2(1-\beta)\tilde{\bm{p}}_{21}^\top \pi_{\bm{y}_{21}} \tilde{\bm{v}}_{21}
\leq (1-\beta)\left( \tilde{\bm{p}}_{21}^\top \pi_{\bm{y}_{21}} \tilde{\bm{p}}_{21}
+ \|\tilde{\bm{v}}_{21}\|^2 \right). 
\end{align*}
Moreover, since $\pi_{\bm{y}_{21}} \preceq I_3$,
$\tilde{\bm{v}}_{21}^\top \pi_{\bm{y}_{21}} \tilde{\bm{v}}_{21} \leq \|\tilde{\bm{v}}_{21}\|^2$.
Hence,
\begin{align} \label{eq:Qpi_final}
\tilde{\bm{x}}_{21}^\top S_\pi \tilde{\bm{x}}_{21} 
&\geq c_p \tilde{\bm{p}}_{21}^\top \pi_{\bm{y}_{21}}\tilde{\bm{p}}_{21} \tilde{\bm{p}}_{21} - c_v \|\tilde{\bm{v}}_{21}\|^2,
\end{align}
where $c_p = 3 \beta - 2$ and $c_v = 2 - \beta$. 

On the other hand, since $\Sigma_2 \geq \frac{1}{\overline{m}_2} I_3$, and provided that $\lambda_2 + 2k_{v_2} > \frac{2}{\lambda_2 + k_{v_2}}$, we have 
\begin{align}
\tilde{\bm{x}}_{21}^\top S_\Sigma \tilde{\bm{x}}_{21}  &\geq \tfrac{\lambda_2}{\overline{m}_2} \|\tilde{\bm{p}}_{21}(\tau)\|^2 + \bar{c}_{v} \|\tilde{\bm{v}}_{21}(\tau)\|^2,
\label{eq:QM_final} 
\end{align}
where $\bar{c}_{v} = \frac{1}{\overline{m}_2}(\lambda_2 + 2 k_{v_2} - \frac{2}{\lambda_2 + k_{v_2}} )$.
Combining \eqref{eq:Qpi_final} and \eqref{eq:QM_final}, we obtain
\begin{align*}
\dot{V}(t) 
&\leq - \tilde{\bm{p}}_{21}^\top (c_p \pi_{\tilde{\bm{p}}_{21}} + \tfrac{\lambda_2}{\overline{m}_2} I_3) \tilde{\bm{p}}_{21}
- (\bar{c}_v - c_v) \|\tilde{\bm{v}}_{21}\|^2, \\
&\leq - \tfrac{\lambda_2}{\overline{m}_2} \| \tilde{\bm{p}}_{21} \|^2
- (\bar{c}_v - c_v) \|\tilde{\bm{v}}_{21}\|^2.
\end{align*}
Thus, under the condition
$\overline{m}_2 < \tfrac{(\lambda_2 + 2 k_{v_2})(\lambda_2 + k_{v_2}) - 2}{2(\lambda_2 + k_{v_2}) - 1}$,
\begin{align} \label{eq:V2_inequality}
    \dot{V}_2(t) \leq - \alpha \|\tilde{\bm{x}}_{21}(t)\|^2, \;\, \alpha = \min\left(\tfrac{\lambda_2}{\overline{m}_2}, \bar{c}_v - c_v\right) > 0.
\end{align}
Since $P_2(t)$ is uniformly bounded and positive definite, 
and using the bounds on $\Sigma_2$, define $\overline{\kappa}_2 := (1+\beta)/\underline{m}_2$, $\underline{\kappa}_2 := (1-\beta)/\overline{m}_2$. This yields the uniform bounds $\underline{\kappa} I_6 \preceq P_2(t) \preceq \overline{\kappa} I_6$.
Consequently,
$\underline{\kappa} \|\tilde{\bm{x}}_{21}(t)\|^2 \leq V_2(t) \leq \overline{\kappa}_2 \|\tilde{\bm{x}}_{21}(t)\|^2$.
Combining this with \eqref{eq:V2_inequality}, we obtain
$
\|\tilde{\bm{x}}_{21}(t)\|
\leq \sqrt{\tfrac{\overline{\kappa}_2}{\underline{\kappa}_2}} e^{-\alpha t/2\overline{\kappa}_2} \|\tilde{\bm{x}}_{21}(0)\|
$.
Hence, the equilibrium $\tilde{\bm{x}}_{21} = 0$ is uniformly exponentially stable.


\noindent
($i=3$): Consider agent $3$. Its neighbor set satisfies $\mathcal{N}_3 \subseteq \{1,2\}$. 
If $\mathcal{N}_3 = \{1\}$, the result follows directly from the case $i=2$. 
Otherwise, if $\mathcal{N}_3 = \{1,2\}$ or $\mathcal{N}_3 = \{1,2\}$,
one has $\tilde{\bm{p}}_{31} = \tilde{\bm{p}}_{32} + \tilde{\bm{p}}_{21}$.
Using \eqref{eq:2nd_err_dynamics} and \eqref{eq:2nd_control_law}, the closed-loop system for $\tilde{\bm{x}}_{3j}$, $j \in \mathcal{N}_3$, can be written as
\begin{align}
\dot{\tilde{\bm{x}}}_{3j} = -A_3(t)\tilde{\bm{x}}_{3j} + B_{21}(t)\tilde{\bm{x}}_{21},
\end{align}
where $A_3(t)$ is defined as in \eqref{eq:xtilde_dynamics} and $B_{21}(t)$ is a bounded matrix-valued function.
This system can be interpreted as a cascaded system with input $\tilde{\bm{x}}_{21}$ to the unforced system
\begin{align} \label{eq:2dyn_3_unforced}
\dot{\tilde{\bm{x}}}_{3j} = -A_3(t)\tilde{\bm{x}}_{3j}.
\end{align}
The proof becomes analogous to the proof of $i=2$ and the equilibrium $\tilde{\bm{x}}_{3j} = 0$ of the unforced system \eqref{eq:2dyn_3_unforced} is uniformly exponentially stable. Moreover, since $B_{21}(t)$ is bounded, it follows that $\tilde{\bm{x}}_{3j}=0$ is uniformly exponentially stable.

\noindent
($i \geq 4$): Assume that for all agents $k \in \{2, \dots, i-1\}$ and all $m \in \mathcal{N}_k$, the equilibrium $\tilde{\bm{x}}_{km} = 0$ is uniformly exponentially stable. 
For agent $i \in \mathcal{V}\backslash\{1\}$ and any $j \in \mathcal{N}_i$, the relative errors can be expressed as
$ \tilde{\bm{x}}_{ij} = \tilde{\bm{x}}_{ik} + \tilde{\bm{x}}_{kj}$, $k \in \mathcal{N}_i$.
Iterating along the graph, $\tilde{\bm{x}}_{ij}$ can be written in terms of error variables $\tilde{\bm{x}}_{km}$ with $2 \le k \le i-1$, $m \in \mathcal{N}_k$.
Thus, the closed-loop dynamics can be expressed as
\begin{align}
\dot{\tilde{\bm{x}}}_{ij} 
= -A_i(t)\tilde{\bm{x}}_{ij} 
+ \sum_{2 \le k \le i-1,\, m \in \mathcal{N}_k} B_{km}(t)\tilde{\bm{x}}_{km},
\end{align}
where $A_i(t)$ is defined as in \eqref{eq:xtilde_dynamics} and $B_{km}(t)$ are bounded matrix-valued functions.
This is again a cascaded system, with inputs $\tilde{\bm{x}}_{km}$. By the induction hypothesis, $\tilde{\bm{x}}_{km} =0 $ is uniformly exponentially stable.
Since the unforced system $\dot{\tilde{\bm{x}}}_{ij} = -A_i(t)\tilde{\bm{x}}_{ij}$
is uniformly exponentially stable (by the same arguments as in the case $i=2$), and the perturbation terms are bounded and vanishing, it follows that $\tilde{\bm{x}}_{ij}=0$ is uniformly exponentially stable.
\end{proof}


    \section{Simulation results}\label{sec:simresults}

This section provides simulation results to illustrate the
effectiveness of the proposed control laws \eqref{eq:control_riccati_general} and \eqref{eq:2nd_control_law}.

\subsection{Formation control for single-integrator motion dynamics}
We compare the performance of the Riccati-based controller~\eqref{eq:control_riccati_general} against the constant-gain law~\eqref{eq:control_general} using a four-agent system in $\mathbb{R}^3$ ($\mathcal{V} = \{1,2,3,4\}$) under the chain leader-follower graph $\mathcal{N}_i = \{i-1\}$ ($i \geq 2$). 
In this scenario, the leader is static at the origin $\bm{p}_1 = [0\; 0\; 0]^\top$. The desired relative positions  rotate about the $z$-axis as $\bm{p}_{i1}^\star(t) = R^\top(t)\bm{p}_{i1}^\star(0)$, where $R(t) \in \SO(3)$ is generated by the angular velocity  $\omega(t) = [0\; 0\; 0.2]^\top$.
The initial configuration is given by $\bm{p}_{2}^\star(0) = [0\;1\;0]^\top$, 
$\bm{p}_{3}^\star(0) = [\frac{\sqrt{3}}{2}\;\frac{1}{2}\;0]^\top$, 
$\bm{p}_{4}^\star(0) = [\frac{1}{2}\;\frac{\sqrt{3}}{2}\;1]^\top$, 
which defines a rotating pyramid-shaped formation that ensures the desired bearings $\bm{y}_{ij}^\star(t)$ are well-excited. 
The initial conditions are $\bm{p}_{3}(0) = [-\frac{1}{2}\;  1\;  \frac{1}{2}]^\top$, $\bm{p}_{2}(0) = [-1\; -\frac{1}{2}\; -\frac{1}{2}]^\top$, $\bm{p}_{4}(0) = [-\frac{1}{4}\; -\frac{1}{4}\;  \frac{1}{2}]^\top$.
The gains are selected as follows: $\forall i \in \mathcal{V}\backslash\{1\}$, $k_i = 1$ for controller \eqref{eq:control_general}, and $M_i(0) = 0.1 I_3$ and $\lambda_i = 0.1$ for controller \eqref{eq:control_general}.
\begin{figure}[h]
    \centering
    \includegraphics[width=.97\linewidth]{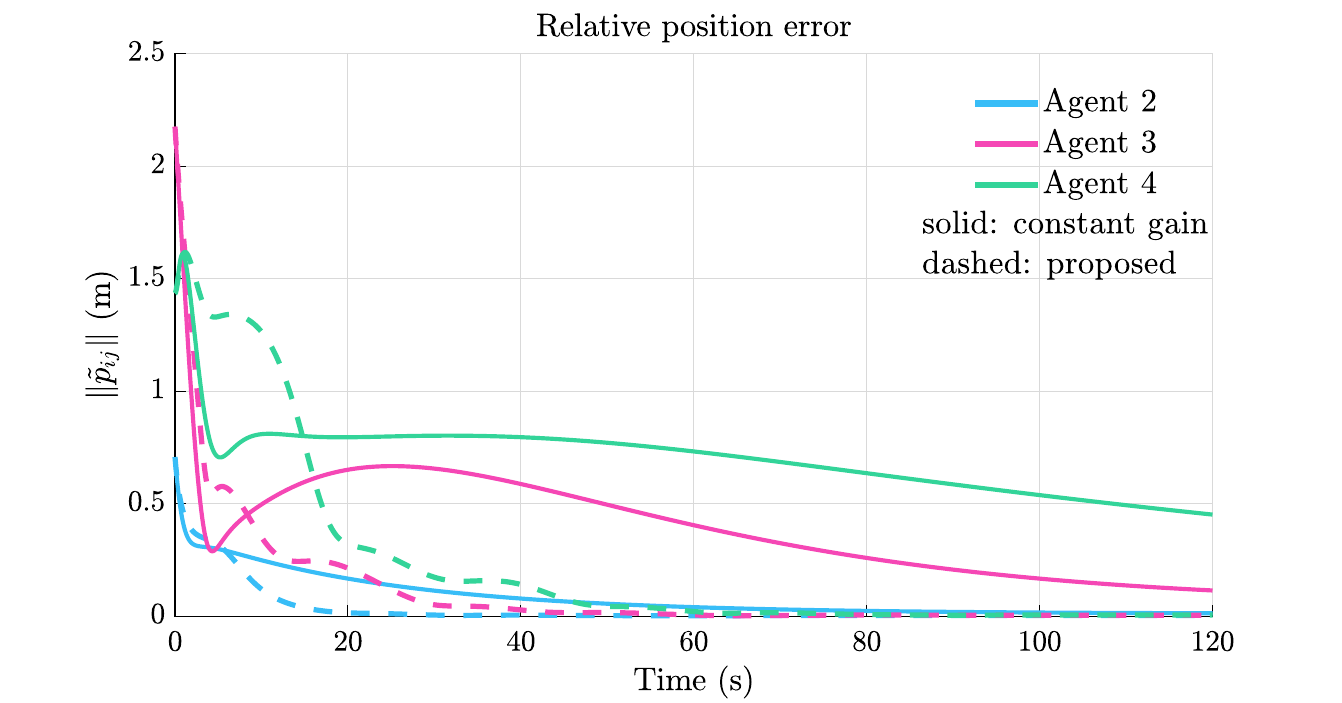}
    \caption{Relative position error norms $\|\tilde{\bm{p}}_{ij}(t)\|$ for the rotating pyramid formationn.}
    \label{fig:sim1}
\end{figure}

Fig. \ref{fig:sim1} shows the agents' relative position error norms $\|\tilde{\bm{p}}_{ij}(t)\|$, the corresponding simulation results with animations of the 3D formation trajectories are available at {\tt\small https://tinyurl.com/2zjzm8ms}.
The errors converge for both controllers, although the fixed gain controller~\eqref{eq:control_general} exhibits significantly slower convergence, particularly for $\tilde{\bm{p}}_{43}$, which is most affected by the chain topology. In contrast, the proposed controller achieves faster and more uniform convergence, consistent with its convergence rate being independent of the PE bound.

\subsection{Formation control for double-integrator motion dynamics}

To compare the control law \eqref{eq:2nd_control_law} with \eqref{eq:control_general_2nd}, we consider a quasi-static stabilization scenario of a group of underwater vehicles, where five followers are required to converge to a desired star formation around a static leader located at $\bm{p}_1^\star = [0\; 0\; 0]^\top$ ($\mathcal{V} = \{1,\dots,6\}$, $\mathcal{N}_i = \{1\}$ for all $i \geq 2$). The desired relative positions evolve along small arcs with angular amplitude $a = \tfrac{\pi}{36}$ and frequency $\omega = 0.5$\,rad/s:
\begin{equation}\label{eq:quasi_static_desired}
    \bm{p}_{i1}^\star(t) 
    = r_i \begin{bmatrix}
        \cos\bigl(\theta_i + a \sin(\omega t)\bigr) \\
        \sin\bigl(\theta_i + a \sin(\omega t)\bigr) \\
        0
    \end{bmatrix}, \;\, \bm{v}_{i1}^\star(t) = \dot{\bm{p}}_{i1}^\star(t),
\end{equation}
with $r_i = 2$ m, and $\theta_i = \frac{(i-2)2\pi}{5}$, $i \in \mathcal{V}\backslash\{1\}$. 
This configuration places the followers at the vertices of a regular pentagon in the $xy$-plane. The angular motion induces small variations in the desired bearings relative to the leader that result in a bearing-PE formation with a small lower bound in the PE condition. 
The gains are as follows: $k_{p_i} = 0.8$ and $k_{v_i} = 1.5$ for controller \eqref{eq:control_general_2nd}, and $M_i(0) = 0.1 I_3$, $\lambda_i = 0.25$, and $k_{v_i} = 1.5$ for all $i \in \mathcal{V}\backslash\{1\}$ for controller \eqref{eq:2nd_control_law}.
\begin{figure}[h]
    \centering
    \includegraphics[width=.95\linewidth]{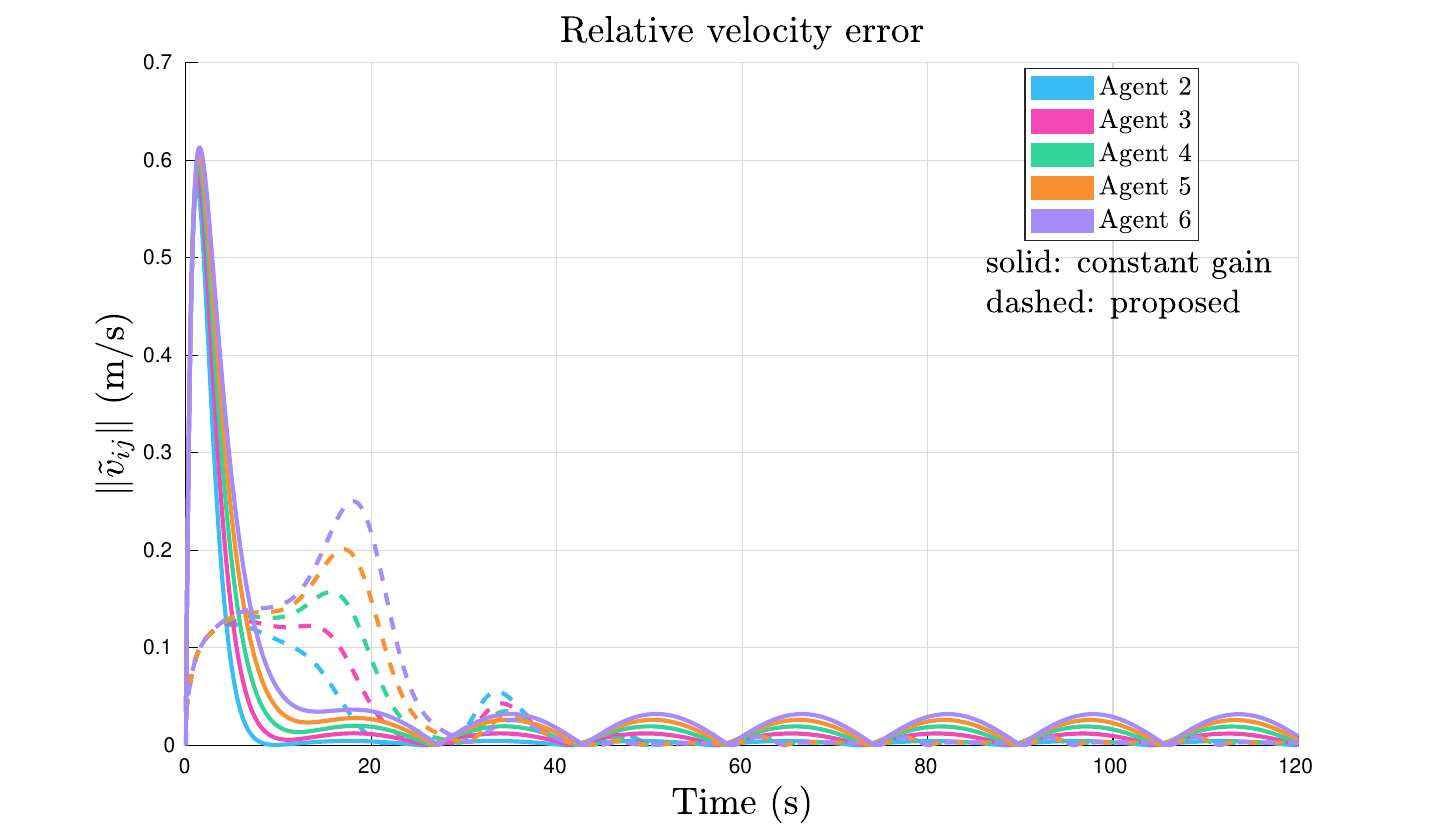}
    \includegraphics[width=.95\linewidth]{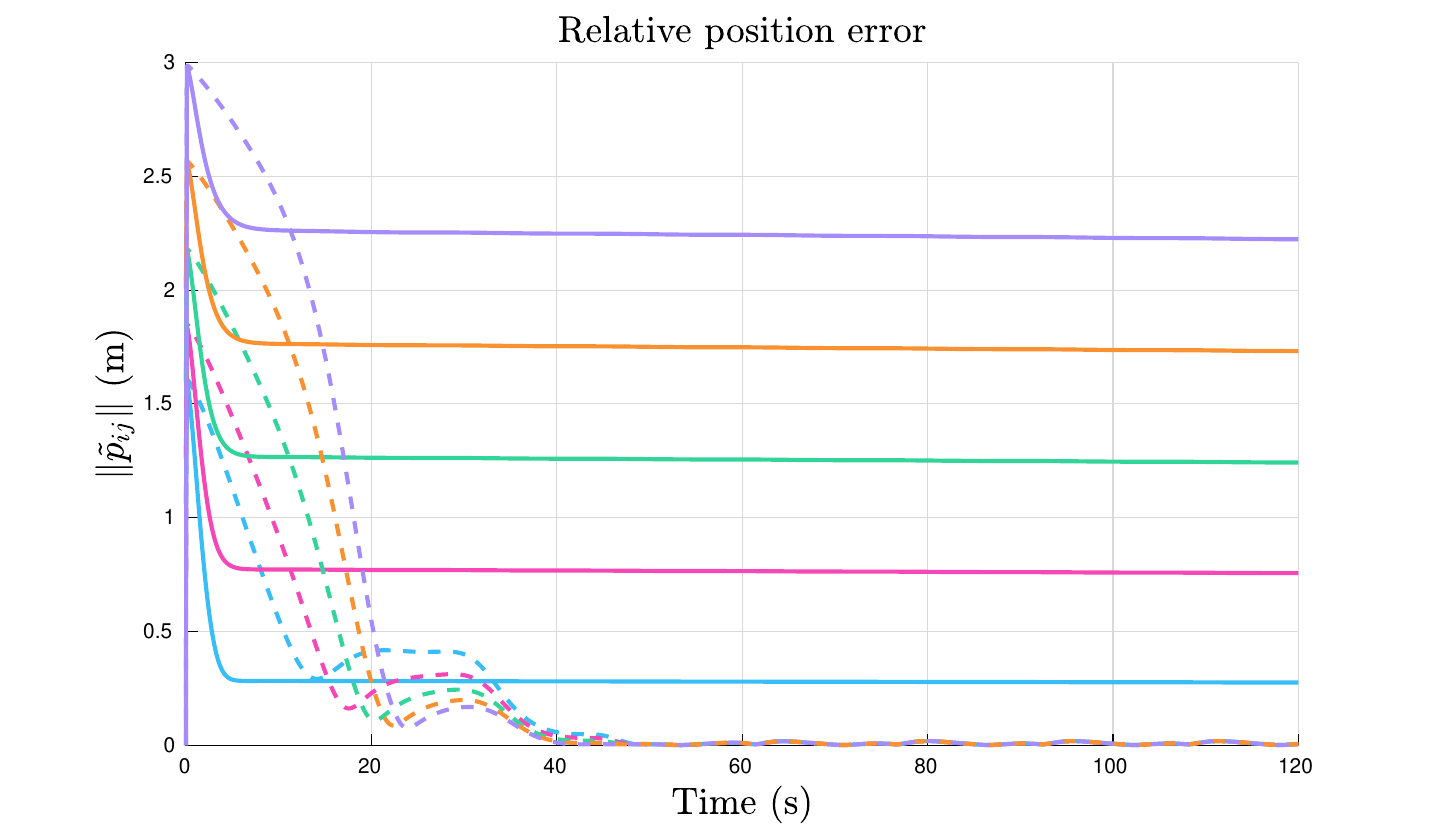}
    \caption{Relative velocity and position error norms $\|\tilde{\bm{v}}_{ij}(t)\|$ and $\|\tilde{\bm{p}}_{ij}(t)\|$ for the star formation.}
    \label{fig:sim2}
\end{figure}

The simulation results are illustrated in Fig. \ref{fig:sim2}, and an animation of the 3D formation trajectories is available at {\tt\small https://tinyurl.com/yc8unad8}. The proposed controller achieves exponential convergence for all agents, whereas the fixed gain controller exhibits significantly slower convergence under the same bearing-PE conditions.

    \section{Conclusions} \label{sec:conclusion}
In this paper, we addressed leader-follower formation control under directed acyclic graph topologies using bearing and relative velocity measurements. By exploiting the accumulated bearing information of the desired formation, we proposed a time-varying gain design whose convergence rate is decoupled from the parameters characterizing the bearing persistence of excitation.
The proposed gain is defined through a Riccati-type equation, with its inverse corresponding to a filtered bearing Gramian. Bearing-PE ensures the well-posedness and uniform boundedness of the gain, while the resulting exponential convergence rate is determined by the gain design parameter. Convergence of the formation error was established under a bearing-PE condition without requiring bearing rigidity assumptions. The framework was further extended to double-integrator dynamics.
Simulation results illustrate the convergence performance and demonstrate the advantage of the proposed gain over classical designs.
Future work will focus on incorporating collision avoidance mechanisms and extending the framework to more general interaction topologies and time-varying formations.

\appendix

\begin{lemma} \label{lem:piy_PE}
    If the matrix $\sum_{j\in\mathcal{N}_i} \pi_{\bm{y}_{ij}^\star}$ is PE in the sense of Definition \ref{def:PE} and $\bm{y}_{ij}(t) \to \bm{y}_{ij}^\star(t)$ exponentially for all $j \in \mathcal{N}_i$. Then, there exist $t_0 >0$ and $\mu' > 0$ such that
    $$\frac{1}{\delta}\int_{t}^{t+\delta} \sum_{j\in\mathcal{N}_i} \pi_{\bm{y}_{ij}(\tau)}d\tau \geq \mu' I_3 > 0,
    $$ 
    for all $t\geq t_0$.
\end{lemma}
The proof is straightforward and is omitted for brevity.

\begin{lemma}[\cite{morin2017uniform}] \label{lem:nonPE_condition}
    Suppose that the matrix-valued function $\pi_{\bm{y}}(t)$ is not persistently exciting. Then, for any $\bar{\delta} > 0$, there exist a sequence $\{t_p\}_{p\in\mathbb{N}}$ and a vector $\bm{z} \in \mathbb{S}^2$ such that
$$\lim_{p\to\infty}\tfrac{1}{\bar{\delta} }\int_{t_p}^{t_p+\bar{\delta} }\bm{z}^\top\pi_{\bm{y}(\tau)}\bm{z}d\tau = 0.$$
Moreover, the vector $\bm{y}(t)$ converges to a constant vector $\bar{\bm{y}} \in \mathbb{S}^2$ as $t \to \infty$.

\end{lemma}

\begin{lemma}[Projected Inverse] \label{lem:projected_inverse}
Let $A_\varepsilon \in \mathbb{R}^{n\times n}$ be a symmetric positive matrix and be invertible for all $\varepsilon \neq 0$, and let  $P$ be an orthogonal projector $P^2 = P$. Assume that in the limit $\varepsilon \to 0$, $A_\varepsilon \to P$ in the sense that the restriction $P A_\varepsilon P$ is invertible on $\mathrm{Im}(P)$ for small $\varepsilon$. Then the projected inverse satisfies $$\lim_{\varepsilon \to 0} P A_\varepsilon^{-1} P = P. $$
Equivalently  $P A_\varepsilon^{-1} P=P$ is the pseudo inverse of $A_\varepsilon$ on $\mathrm{Im}(P)$.

\end{lemma}

\begin{proof}
Let $\mathbb{R}^n = \mathrm{Im}(P) \oplus \ker(P)$ and denote by 
$B_\varepsilon := P A_\varepsilon P : \mathrm{Im}(P) \to \mathrm{Im}(P)$ the restriction of $A_\varepsilon$ to $\mathrm{Im}(P)$. By assumption, $B_\varepsilon$ is invertible for small $\varepsilon$ and $B_\varepsilon \to P|_{\mathrm{Im}(P)} = I_r$, with $r = \dim(\mathrm{Im}(P))$.

Now, let $\bm{x} \in \mathrm{Im}(P)$ and define $\bm{y} := P A_\varepsilon^{-1} P \bm{x} \in \mathrm{Im}(P)$. Since $\bm{y} = P \bm{y}$, we can rewrite
$$
B_\varepsilon \bm{y} = P A_\varepsilon \bm{y} = P A_\varepsilon P A_\varepsilon^{-1} P \bm{x}.
$$
Using $A_\varepsilon \to P$ together with $P^2 = P$, we obtain $P A_\varepsilon P \to P$, and consequently $B_\varepsilon \bm{y} \to P \bm{x} = \bm{x}$.
Finally, since $B_\varepsilon$ is invertible for small $\varepsilon$ and $B_\varepsilon \to I_r$, we deduce $\bm{y} = B_\varepsilon^{-1} \bm{x} \to \bm{x}$. Hence $P A_\varepsilon^{-1} P \bm{x} \to \bm{x}$ for all $\bm{x} \in \mathrm{Im}(P)$, which implies $\lim_{\varepsilon \to 0} P A_\varepsilon^{-1} P = P$.
\end{proof}

\bibliography{ref}

\end{document}